\documentclass[english,12pt]{article}
\usepackage{subfiles}
\usepackage{xspace}
\newcommand{\ea}{\emph{et al.}\xspace}
\newcommand{\ph}{\ensuremath{\overline{p_{j}}}\xspace}
\newcommand{\pb}{\ensuremath{\underline{p_{j}}}\xspace}
\newcommand{\phm}{\overline{p_{j}}}
\newcommand{\pbm}{\underline{p_{j}}}
\newcommand{\xrm}{x_{j}^{R}\xspace}
\newcommand{\xam}{x_{j}^{A}\xspace}
\newcommand{\xr}{$x_{j}^{R}$\xspace}
\newcommand{\xa}{$x_{j}^{A}$\xspace}
\newcommand{\algname}{\textsc{HLP-$b$}\xspace}
\newcommand{\pbname}{$(Pm,Pk)\mid \mathit{prec} \mid C_{\max}$\xspace}

\usepackage{algorithm}
\usepackage[noend]{algorithmic}
\usepackage{amssymb}
\usepackage{amsmath} % erreur sur la redefinition de \vec qui passe de bold a pas bold...
\usepackage{paralist}
\usepackage{todonotes}

\usepackage{fancyhdr}

\usepackage{a4}

\providecommand{\keywords}[1]{ \ \\ \small \textbf{\textit{Keywords---}} #1}

\newcommand\tab[1][1cm]{\hspace*{#1}}

\usepackage{amsmath, amssymb, amsthm} % Mathematical packages
\usepackage{cleveref}
\newtheorem{lemma}{Lemma}
\newtheorem{theorem}{Theorem}

\newtheorem{corollary}{Corollary}[theorem]
\newtheorem{hypo}{Hypothesis}

\usepackage{setspace}

\author{Vincent Fagnon\textsuperscript{1,2}, Imed Kacem\textsuperscript{2}, Giorgio Lucarelli\textsuperscript{2} \\ and Bertrand Simon\textsuperscript{3}}

\date{}

\begin{document}

\title{Scheduling on Hybrid Platforms: Improved Approximability Window}
\maketitle
\makeatletter{
\renewcommand*{\@makefnmark}{}

\begin{abstract}
Modern platforms are using accelerators in conjunction with standard processing units in order to reduce the running time of specific operations, such as matrix operations, and improve their performance.
Scheduling on such hybrid platforms is a challenging problem since the algorithms used for the case of homogeneous resources do not adapt well.
In this paper we consider the problem of scheduling a set of tasks subject to precedence constraints on hybrid platforms, composed of two types of processing units.
We propose a $(3+2\sqrt{2})$-approximation algorithm and a conditional lower bound of 3 on the approximation ratio.
These results improve upon the 6-approximation algorithm proposed by Kedad-Sidhoum \ea as well as the lower bound of 2 due to Svensson for identical machines.
Our algorithm is inspired by the former one and distinguishes the allocation and the scheduling phases.
However, we propose a different allocation procedure which, although is less efficient for the allocation sub-problem, leads to an improved approximation ratio for the whole scheduling problem.
This approximation ratio actually decreases when the number of processing units of each type is close and matches the conditional lower bound when they are equal.

\end{abstract}

\footnotetext{\textsuperscript{1} Univ. Grenoble Alpes, CNRS, Inria, Grenoble INP, LIG, 38000 Grenoble, France \\ \tab \tab vincent.fagnon@univ-grenoble-alpes.fr}
\footnotetext{\textsuperscript{2} LCOMS, University of Lorraine, Metz, France \\ \tab \tab \{first\_name.last\_name\}@univ-lorraine.fr }
\footnotetext{\textsuperscript{3} Universität Bremen, Bremen, Germany \\ \tab \tab bsimon@uni-bremen.de}}
\newpage

%%%%%%%%%%%%%%%%%%%%%%%%%%%%%%%%%%%%%%%%%%%%%%%%%%%%%%%%%%%%%%%%%%%%%%%%%%%%
%%%%%%%%%%%%%%%%%%%%%%%%%%%%%%%%%%%%%%%%%%%%%%%%%%%%%%%%%%%%%%%%%%%%%%%%%%%%
%%%%%%%%%%%%%%%%%%%%%%%%%%%%%%%%%%%%%%%%%%%%%%%%%%%%%%%%%%%%%%%%%%%%%%%%%%%%

\section{Introduction}

Nowadays, more and more High Performance Computing platforms use special purpose processors in conjunction with classical Central Processing Units (CPUs) in order to accelerate specific operations and improve their performance. A typical example is the use of modern Graphics Processing Units (GPUs) which can accelerate vector and matrix operations.

Due to the heterogeneity that introduce this kind of accelerators, the scheduling problem on such hybrid platforms becomes more challenging. Several experimental results as well as theoretical lower bounds \cite{online4} show that the decision of the allocation of a task to the type of processors is crucial for the performance of the system. Specifically, classical greedy policies, such as Graham's List Scheduling \cite{Graham69}, which perform well in the case of identical computing resources, fail to generalize on hybrid platforms. For this reason, all known algorithms for hybrid platforms \cite{online4,online2,ChenYZ14,Kedad-Sidhoum} choose the type of the resource for each task before deciding its scheduling in the time horizon.

In this paper, we focus on the problem of scheduling an application on such an hybrid platform consisting of $m$ identical CPUs and $k$ identical GPUs.
An application is described as a set of $n$ mono-processor tasks $V$ which are linked through precedence dependencies described by a directed acyclic graph $G=(V,E)$. This means that a task can start being executed only after all of its predecessors are completed. The processing time of task $j$ on a CPU (resp. on a GPU) is denoted by \ph (resp. by \pb), and we do not assume any relation between \ph and \pb. This is justified in real systems where tasks performing for instance matrix operations can be executed much more efficiently on a GPU, while the execution of tasks which need to communicate often with the file system is faster on a CPU. Therefore, we can assume without loss of generality than $m\geq k$.

We are interested in designing polynomial-time algorithms with good performance guarantees in the worst case. As performance measure we use the well-known approximation ratio which compares the solution of an algorithm and the optimal solution with respect to an objective function. In this paper, we study the \emph{makespan} objective, that is we aim at minimizing the completion time of the last task. Extending the Graham notation, we will denote this problem as \pbname.

For this problem, a 6-approximation algorithm named HLP (Heterogeneous Linear Program) has been proposed by Kedad-Sidhoum \ea \cite{Kedad-Sidhoum}. This algorithm has two phases. In the first phase a ``good'' allocation of each task either on the CPU or on the GPU side is decided. This decision is based on an integer linear program which uses a 0-1 decision variable $x_j$ for each task $j$: $x_j$ will be  equal to one if $j$ is assigned to the CPU side, and to zero otherwise. This integer linear program does not model the whole scheduling problem but only the allocation decision, trying to balance the average load on the CPUs and GPUs as well as the critical path length. The fractional relaxation of this program is solved and the allocation of each task $j$ is determined by a simple rounding rule: it is assigned to GPUs if $x_j<1/2$, and to CPUs otherwise. In the second phase, the greedy List Scheduling algorithm is used to schedule the tasks respecting the precedence constraints and the allocation defined in the first phase.

The authors in \cite{Kedad-Sidhoum} prove that the value of $1/2$ chosen is best possible with respect to the linear program used in the first phase. In a sense, they prove that the integrality gap of the linear program relaxation is 2. Furthermore, given this simple rounding rule based on $1/2$, Amaris \ea \cite{online4} present a tight example of HLP which asymptotically attains an approximation ratio of 6, even if another scheduling algorithm is used in the second phase. Despite both previous negative results, we show that HLP can achieve a better approximation ratio by using a different rounding procedure. Indeed, even though we use a rounding which is not the best possible with respect to the allocation problem solved in the first phase, this rounding allows us to obtain stronger guarantees on the scheduling phase and therefore improve the approximation ratio. 
The main difference with HLP is that we allocate task $j$ to the fastest processor type if $x_j$ is close to $1/2$ in the fractional relaxation solution. We then achieve an approximation ratio smaller than $3+2\sqrt{2}$ and that tends towards 3 when $m/k$ is close to $1$.

The best known lower bound on the approximation ratio is the same as for identical machines, {\it i.e.}, $4/3$~\cite{lenstra1978}, but can  be improved to 2 by assuming a variant of the unique games conjecture~\cite{svensson}. Our second contribution is to improve this conditional lower bound to 3 for any value of $m/k$ assuming a stronger conjecture introduced by Bazzi and Norouzi-Fard~\cite{bazzi}. 
This conditional lower bound is therefore tight when $m=k$.

\subsection*{Organization of the paper}

In Section~\ref{sec:related} we give a literature review by positioning our problem with respect to closely related ones and by presenting several known approximability results.
In Section~\ref{sec:algo} we present our adapted algorithm for the problem of scheduling on hybrid platforms as well as its analysis which leads to an approximation ratio of 5.83. In \Cref{sec:inapprox}, we prove a conditional lower bound of 3 on the approximation ratio.
Finally, we conclude in Section~\ref{sec:conclusion}.

%%%%%%%%%%%%%%%%%%%%%%%%%%%%%%%%%%%%%%%%%%%%%%%%%%%%%%%%%%%%%%%%%%%%%%%%%%%%
%%%%%%%%%%%%%%%%%%%%%%%%%%%%%%%%%%%%%%%%%%%%%%%%%%%%%%%%%%%%%%%%%%%%%%%%%%%%
%%%%%%%%%%%%%%%%%%%%%%%%%%%%%%%%%%%%%%%%%%%%%%%%%%%%%%%%%%%%%%%%%%%%%%%%%%%%

\section{Related Work}
\label{sec:related}

The problem of scheduling on hybrid platforms consisting of two sets of identical processors is a generalization of the classical problem of scheduling on parallel identical processors, denoted by $P \mid prec \mid C_{\max}$.
On the other hand, our problem is a special case of the problem of scheduling on unrelated processors (denoted by $R \mid prec \mid C_{\max}$), where each task has a different processing time on each processor.
Moreover, in the case of scheduling on related processors (denoted by $Q \mid prec \mid C_{\max}$), each processor has its specific speed and the processing time of each task depends on the speed of the assigned processor.
This problem is more general than $P \mid prec \mid C_{\max}$ in the sense that the processing time of a task is different on each processor.
However, in the former problem all tasks are accelerated or decelerated by the same factor when using a specific processor, while in our case two tasks does not necessarily have the same behavior (acceleration or deceleration) if they are scheduled on a CPU or a GPU.

For $P \mid prec \mid C_{\max}$, the greedy List Scheduling algorithm proposed by Graham~\cite{Graham69} achieves an approximation ratio of $(2-\frac{1}{m})$, where $m$ is the number of the processors.
Svensson~\cite{svensson} proved that this is the best possible approximation result that we can expect, assuming $P\neq NP$ and a variant of the unique games conjecture introduced by Bansal and Khot~\cite{bansal2009}.
Note that this negative result holds also for our more general problem.
For $Q \mid prec \mid C_{\max}$, a series of algorithms with logarithmic approximation ratios are known (see for example~\cite{Cheruki,ChudakShmoys}), while Li~\cite{shili} has recently proposed a $O(\log(m)/\log(\log(m)))$-approximation algorithm which is the current best known ratio.
On the negative side, Bazzi and Norouzi-Fard~\cite{bazzi} show that it is not possible to have a constant approximation ratio assuming the NP-hardness of some problems on $k$-partite graphs.
No result is actually known for $R \mid prec \mid C_{\max}$.
However, there are few approximation algorithms for special classes of precedence graphs (see for example~\cite{R-tree}).

For the problem \pbname, targeting hybrid platforms, Kedad-Sidhoum \ea \cite{Kedad-Sidhoum} presented a 6-approximation algorithm as we reported before by separating the allocation and the scheduling phases.
Amaris \ea \cite{online4} proposed small improvements on both phases, without improving upon the approximation ratio.
However, they show that using the rounding proposed in \cite{Kedad-Sidhoum}, any scheduling policy cannot lead to an approximation ratio strictly smaller than 6.
In the absence of precedence constraints, a polynomial time approximation scheme has been proposed by Bleuse \ea \cite{bleuse6}.

The problem of scheduling on hybrid platforms has been also studied in the online case.
If the tasks are not subject to precedence relations, then a 3.85-competitive algorithm has been proposed in~\cite{ChenYZ14}, while the authors show also that no online algorithm can have a competitive ratio strictly less than 2.
In the presence of precedence constraints, Amaris \ea \cite{online4} consider that tasks arrive in an online order respecting the precedence relations and they give a $(4\sqrt{m/k})$-competitive algorithm.
This result has been improved by Canon \ea \cite{online2} who provide a $(2\sqrt{m/k}+1)$-competitive algorithm, while they show that no online algorithm can have a competitive ratio smaller than $\sqrt{m/k}$.

%%%%%%%%%%%%%%%%%%%%%%%%%%%%%%%%%%%%%%%%%%%%%%%%%%%%%%%%%%%%%%%%%%%%%%%%%%%%
%%%%%%%%%%%%%%%%%%%%%%%%%%%%%%%%%%%%%%%%%%%%%%%%%%%%%%%%%%%%%%%%%%%%%%%%%%%%
%%%%%%%%%%%%%%%%%%%%%%%%%%%%%%%%%%%%%%%%%%%%%%%%%%%%%%%%%%%%%%%%%%%%%%%%%%%%

\section{A 5.83-approximation Algorithm}
\label{sec:algo}
In this section we present the improved approximation algorithm and its analysis for the problem \pbname.
Although several ingredients of our algorithm have been already presented in \cite{Kedad-Sidhoum}, we present here all the steps of the algorithm for the sake of completeness.

%%%%%%%%%%%%%%%%%%%%%%%%%%%%%%%%%%%%%%%%%%%%%%%%%%%%%%%%%%%%%%%%%%%%%%%%%%%%

\subsection{The Algorithm \algname}

As explained in introduction, the algorithm \algname has two phases: the allocation phase and the scheduling one.
The allocation phase is based on an integer linear program.
For each task $j \in V$, let $x_j$ be a decision variable which is equal to 1 if task $j$ is assigned to the CPU side, and to 0 otherwise.
Moreover, let $C_j$ be a variable corresponding to the completion time of task $j$.
Finally, let $C_{max}$ be a variable that indicates the maximum completion time over all tasks.
For the sake of simplicity, we add in $G$ a fictive task 0 with $\overline{p_0}=\underline{p_0}=0$ which precedes all other tasks.
Consider the following integer linear program similarly to Kedad-Sidhoum \ea \cite{Kedad-Sidhoum}.

\begin{align}
\textrm{Minimize }& C_{\max} \nonumber\\
\frac{1}{m} \sum_{j\in{V}} \phm x_{j} &\leq C_{max}  \label{c:loadCPU}\\
\frac{1}{k} \sum_{j\in{V}} \pbm (1-x_{j}) &\leq C_{max}  \label{c:loadGPU} \\
C_{i}+\phm x_{j}+\pbm(1-x_{j})&\leq C_{j} & \forall (i,j)\in E \label{c:CP}\\
0  \leq  C_{j}  &\leq  C_{max} & \forall j \in V \label{c:Cmax}\\
x_{j} &\in \{0,1\} & \forall j \in V \label{c:int}
\end{align}

Constraints~(\ref{c:loadCPU}) and~(\ref{c:loadGPU}) imply that the makespan of any schedule cannot be smaller than the average \emph{load} on the CPU and GPU sides, respectively.
Constraints~(\ref{c:CP}) and~(\ref{c:Cmax}) build up the \emph{critical path} of the precedence graph, i.e., the path of $G$ with the longest total completion time.
In any schedule, the critical path length is a lower bound of the makespan.
Note that the critical path of the input instance cannot be defined before the allocation decision for all tasks since the exact processing time of a task depends on this allocation.
Constraint~\ref{c:int} is the integrality constraint for the decision variable $x_j$.
In what follows, we relax the integrality constraint and we replace it by $x_j \in [0,1]$ for each task $j$ in $V$, in order to get a linear program which we can solve in polynomial time.
The above integer linear program is not completely equivalent to our scheduling problem, but the objective value of its optimal solution is a lower bound of any optimal schedule.

The rounding procedure of \algname is based on a parameter $b\geq 2$. We will show in \Cref{subsec:analysisalgo} that the best choice is $b= 1+\sqrt{\frac{2-k/m}{1-k/m}}$.
Let \xr be the value of the decision variable for task $j$ in an optimal solution of the above linear program relaxation.
We define \xa to be the value of the decision variable for task $j$ in our algorithm's schedule, that is the value of the decision variable obtained by the rounding procedure.
The allocation phase of our algorithm rounds the optimal relaxed solution $\{\xrm\}$ to the feasible solution $\{\xam\}$ as follows:
\begin{itemize}
\item if $\xrm \geq 1- \frac{1}{b}$, then $\xam = 1$;
\item if $\xrm \leq \frac{1}{b}$, then $\xam = 0$;
\item if $\frac{1}{b} < \xrm < 1-\frac{1}{b}$ and $\overline{p_j} \geq \underline{p_j} $, then $\xam = 0$;
\item if $\frac{1}{b} < \xrm < 1-\frac{1}{b}$ and $\overline{p_j} < \underline{p_j} $, then $\xam = 1$.
\end{itemize}
Intuitively, if the linear program solution is close to an integer ($x_j \leq \frac{1}{b}$ or $x_j \geq 1-\frac{1}{b}$) then we follow its proposal, else we choose the processor type with the smallest processing time: the task is allocated to a CPU ({\it i.e.}, $\xam=1$),  if $\overline{p_j} < \underline{p_j}$ and to a GPU otherwise.

Given the allocation obtained by the previous procedure, \algname proceeds to the scheduling phase.
The classical List Scheduling algorithm is applied respecting the allocation $\{\xam\}$ and the precedence constraints: tasks are allocated to the earliest available processor of the correct type in a topological order.

%%%%%%%%%%%%%%%%%%%%%%%%%%%%%%%%%%%%%%%%%%%%%%%%%%%%%%%%%%%%%%%%%%%%%%%%%%%%

\subsection{Analysis of the Algorithm \algname}
\label{subsec:analysisalgo}

We begin the analysis of \algname with some lemmas that are based on the rounding procedure.

\begin{lemma} \label{lemma:2}
For each task $j\in V$ we have $(1-\xam)\pbm\leq b \cdot (1-\xrm)\pbm$. 
\end{lemma}
\begin{proof}
Consider any task $j \in V$.
Note first that if $j$ is assigned to the CPU side by the algorithm then $\xam=1$ and the lemma directly holds since $\xrm \leq 1$.
Then, we assume that $j$ is assigned to the GPU side, that is $\xam=0$.
Hence, $\xrm \leq (1-\frac{1}{b})$.%, since by definition we have that $b \geq 2$ and thus $\frac{1}{b} \leq (1-\frac{1}{b})$.
Therefore, we conclude as $b \cdot (1-\xrm)\pbm \geq \pbm = (1-\xam)\pbm$.
\end{proof}

\begin{lemma} \label{lemma:3}
For each task $j \in V$ we have:
\begin{equation*}
\xam\phm + (1-\xam)\pbm \leq \frac{b}{b-1} (\xrm\phm + (1-\xrm) \pbm).
\end{equation*}
\end{lemma}
\begin{proof}
Consider any task $j \in V$.
We have the following three cases.

\begin{itemize}
\item If $\xrm \leq \frac{1}{b}$, then $\xam = 0$ and we have:
\begin{align*}
(1-\xrm)\pbm &\geq (1-\frac{1}{b}) (1-\xam)\pbm% = (1-\frac{1}{b}) (1-\xam)\pbm + (1-\frac{1}{b})\xam\phm
\\
(1-\xrm)\pbm + \xrm\phm &\geq (1-\frac{1}{b}) \left((1-\xam)\pbm + \xam\phm\right).
\end{align*}

\item If $\xrm \geq \big( 1-\frac{1}{b} \big)$, then $\xam = 1$ and we have:
\begin{align*}
\xrm \phm &\geq (1-\frac{1}{b}) \xam\phm %= (1-\frac{1}{b}) (1-\xam)\pbm + (1-\frac{1}{b})\xam\phm
\\
(1-\xrm)\pbm + \xrm\phm &\geq (1-\frac{1}{b}) \left((1-\xam)\pbm + \xam\phm\right).
\end{align*}

\item If $\frac{1}{b} < \xrm < (1-\frac{1}{b})$, then we have:\\
  $$\xrm\phm + (1-\xrm)\pbm \geq \min(\phm,\pbm) 
 = \xam\phm + (1-\xam)\pbm.$$
\end{itemize}
Therefore, combining the three cases, we obtain the lemma as $b/(b-1)\geq 1$.
\end{proof}

Based on the three previous lemmas, the following theorem gives the approximation ratio of our algorithm \algname.

\begin{theorem}\label{thm:approx}
\algname achieves an approximation ratio of  $3+4\sqrt{\frac{1-k/m}{2-k/m}}$, which is upper bounded by $3+2\sqrt{2}\leq 5.83$.\end{theorem}
\begin{proof}
We first define some additional notations.
In the algorithm's schedule, let $W_{CPU}^A$ (resp. $W_{GPU}^A$) be the total load over all CPUs (resp. GPUs), and let $CP^A$ be the value of the critical path length of $G$ after the allocation phase of \algname. Denoting by $\mathcal P$ the set of paths in $G$, these values equal:
\begin{align*}
&W_{CPU}^A=\sum_{ j \in V}\, \phm \xam; \qquad W_{GPU}^A=\sum_{ j \in V}\,\pbm(1-\xam); \\ 
&CP^A = \max_{p \in \mathcal{P}} \Big\{\sum_{j \in p}\, (\phm \xam + \pbm (1-\xam))\Big\}.
\end{align*}

In a similar way, we define $W_{CPU}^R$, $W_{GPU}^R$ and $CP^R$ as the total load on CPUs, the total load on GPUs and the critical path in an optimal solution of the linear program relaxation.
Furthermore, let $C_{\max}^A$, $C_{\max}^R$ and $C_{\max}^*$ be respectively the makespan of the schedule created by \algname, the objective value in an optimal solution of the linear program relaxation and the makespan of an optimal solution for our problem.
Following the same arguments as in \cite{online4,Kedad-Sidhoum}, as \algname is a List Scheduling algorithm, we have:

\begin{align*}
C_{max}^A & \leq  \frac{W_{CPU}^{A}}{m} + \frac{W_{GPU}^{A}}{k} + CP^{A}\\
& =  \frac{W_{CPU}^{A}+W_{GPU}^{A}}{m} +\frac{m-k}{mk} W_{GPU}^{A} + CP^{A}\\
& = \frac{1}{m} \sum_{j \in V} \left( \xam\phm + (1-\xam)\pbm  \right) + \frac{m-k}{mk} \sum_{j \in V} (1-\xam)\pbm \\
& \qquad\qquad + \max_{p \in \mathcal{P}} \Big\{\sum_{j \in p}\left( \xam\phm +  (1-\xam)\pbm\right)\Big\}
\end{align*}

Using \Cref{lemma:2,lemma:3}, we obtain:
\begin{align*}
C_{max}^A & \leq  \frac{b}{b-1} \frac{1}{m} \sum_{j \in V} \left( \xrm\phm + (1-\xrm) \pbm \right) + b \frac{m-k}{mk} \sum_{j \in V} (1-\xrm)\pbm \\
& \qquad\qquad + \frac{b}{b-1} \max_{p \in \mathcal{P}} \Big\{\sum_{j \in p}\left( \xrm\phm + (1-\xrm)\pbm\right)\Big\}\\
& =  \frac{b}{b-1} \frac{W_{CPU}^{R}+W_{GPU}^{R}}{m} + b \frac{m-k}{mk} W_{GPU}^{R} + \frac{b}{b-1} CP^{R}
\end{align*}

Now, the constraints~(\ref{c:loadCPU}) to~(\ref{c:Cmax}) of the linear program relaxation give us:
\begin{align*}
C_{max}^A & \leq  \frac{b}{b-1} \frac{mC_{max}^{R}+kC_{max}^{R}}{m} +b \frac{m-k}{mk} kC_{max}^{R} + \frac{b}{b-1} C_{max}^{R}.
\end{align*}

Since $C_{\max}^R \leq C_{\max}^*$ we get:
\begin{align*}
\frac{C_{max}^A}{C_{max}^{*}}
& \leq  \frac{b}{b-1} \cdot \frac{m+k}{m} + b \cdot \frac{m-k}{m} + \frac{b}{b-1}
 =  b + 2 \cdot \frac{b}{b-1}- \frac{m}{k}(b-\frac{b}{b-1}).
\end{align*}
This function reaches its minimum for $b = 1+\sqrt{\frac{2-k/m}{1-k/m}} > 1+\sqrt 2$, which gives: 
\begin{align*}
\frac{C_{max}^A}{C_{max}^{*}} \leq 3+4\sqrt{\frac{1-k/m}{2-k/m}} \leq 3+2\sqrt{2} \approx 5.83.
\end{align*}
\end{proof}

%%%%%%%%%%%%%%%%%%%%%%%%%%%%%%%%%%%%%%%%%%%%%%%%%%%%%%%%%%%%%%%%%%%%%%%%%%%%
%%%%%%%%%%%%%%%%%%%%%%%%%%%%%%%%%%%%%%%%%%%%%%%%%%%%%%%%%%%%%%%%%%%%%%%%%%%%
%%%%%%%%%%%%%%%%%%%%%%%%%%%%%%%%%%%%%%%%%%%%%%%%%%%%%%%%%%%%%%%%%%%%%%%%%%%%

\section{Conditional lower bound on the approximation factor}
\label{sec:inapprox}
 
\newcommand{\mc}[1]{\ensuremath{\mathcal{#1}}\xspace}
\newcommand{\myk}{q\xspace}
 
In this section, we extend the results of Bazzi and Norouzi-Fard~\cite{bazzi} in our setting. Assuming \Cref{hyp} (see below), they show that it is NP-hard to approximate $Q \mid prec \mid C_{\max}$ within a constant factor. If we focus on only two types of processors, their result implies a lower bound of 2 on the approximation ratio, therefore not improving on Svensson's result~\cite{svensson}. We improve their result to obtain a conditional lower bound of 3 stated in \Cref{th:LB}, which therefore also holds in our more restricted setting \pbname in which the processing times on both processor types can be arbitrary. Due to lack of space, we do not discuss further the relevance of \Cref{hyp} or its link to the weaker Unique Games Conjecture and refer the reader to \cite{bazzi} for more details. 

\begin{theorem}
	\label{th:LB} Assuming Hypothesis 1 and $P\neq NP$, there exist no polynomial-time
$(3-\alpha)$-approximation, for any $\alpha>0$, for the problem \pbname, even if the processors are related.
\end{theorem}

%\paragraph*{Hypothesis 1 ($q$-partite problem)}
\begin{hypo}[$q$-partite problem]
\label{hyp}
For every small $\varepsilon,\delta>0$, and every constant integers $\myk, Q>1$, the following problem is NP-hard: given a $\myk$-partite graph $G_q=(V_1,\dots, V_{\myk},E_1,\dots,E_{\myk-1})$ with $|V_i| = n$ for all $1\leq i \leq \myk$ and $E_i$ being the set of edges between $V_i$ and $V_{i+1}$ for all $1\leq i < \myk$, distinguish between the two following cases:
\begin{compactitem}
	\item YES Case: every $V_i$ can be partitioned into $V_{i,0},\dots V_{i,Q-1}$, such that:
	\begin{compactitem}
		\item there is no edge between $V_{i,j_1}$ and $V_{i+1,j_2}$ for all $1\leq i< \myk$, $j_1>j_2$.
		\item $|V_{i,j}| \geq \frac{1-\varepsilon}{Q}n$, for all $1 \leq i \leq \myk$, $0 \leq j \leq Q-1$.
	\end{compactitem}
	\item NO Case: for every $1 \leq i < \myk$ and every two sets $S\subseteq V_i$, $T\subseteq V_{i+1}$ such that $|S|=|T|=\lfloor\delta n\rfloor$, there is an edge between $S$ and $T$.
\end{compactitem}
\end{hypo}

We start by fixing several values: an integer $\myk$ multiple of $3$, an integer $Q$, $\delta \leq 1/(2Q)$ and $\varepsilon\leq1/Q^2$.	
We consider the $\myk$-partite problem parameterized by $Q,\varepsilon,\delta$, which is assumed to be NP-hard under \Cref{hyp}.

\paragraph*{Reduction.} We define a reduction from $G_q=(V_1,\dots, V_{\myk},E_1,\dots,E_{\myk-1})$, a $\myk$-partite graph where for each $i$, $|V_i|=n>Q$, to a scheduling instance $\mc I$. The instance consists of $m=\left\lceil(1+Q\varepsilon)n^4\right\rceil$ CPUs and
$k=\left\lceil(1+Q\varepsilon)n^2\right\rceil$ GPUs and uses two types of tasks: CPU tasks verifying $\ph = n\pb = 1$, and GPU tasks verifying $\ph = n\pb = n$. The tasks and edges ({\it i.e.}, precedence constraints) are defined as follows. 
For each $0\leq z < \myk/3$, and for each:
\begin{compactitem}
	\item vertex $v\in V_{3z+1}$,
create a set $\mc J_{3z+1,v}$ of $Qn-Q$ GPU tasks (type $a$).

	\item vertex $v\in V_{3z+2}$,
create a set $\mc J_{3z+2,v}$ of $Qn^3$ CPU tasks (type $b$).

	\item vertex $v\in V_{3z+3}$,
create a set $\mc J_{3z+3,v}$ of $Q-2$ GPU tasks (type $c$) indexed
$J_{3z+3,v}^1,\dots J_{3z+3,v}^{Q-2}$,
	and an edge from $J_{3z+3,v}^\ell$ to $J_{3z+3,v}^{\ell+1}$ for $\ell$ from 1 to $Q-3$.

	\item edge $(v,w)\in E_i$, create all edges from the set $\mc J_{i,v}$ to the set $\mc J_{i+1,w}$.
\end{compactitem}

Intuitively, the tasks corresponding to each set $V_i$ of $G_{\myk}$ can be computed in $Q$ time slots. To achieve this, each set of type $b$ requires almost all the CPUs, each set of type $a$ requires almost all but $n$ GPUs, and each set of type $c$ requires $n$ GPUs. On a YES instance, it is possible to progress simultaneously on the tasks corresponding to three consecutive sets $V_i$, by pipe-lining the execution, thus obtaining a makespan close to $qQ/3$. For example, it is possible to execute $V_{i,1}$ at some time step, and then to execute $V_{i+1,1}$ and $V_{i,2}$ in parallel. On a NO instance, the tasks corresponding to each $V_i$ have to be scheduled almost independently, thus not efficiently using the processing power: there are too few GPUs to process a significant amount of CPU tasks, and CPUs are too slow to process GPU tasks. The minimum possible makespan is then close to $qQ$. The two following lemmas state these results formally.

\begin{lemma}[Completeness]
	\label{lem:complete}
	If $G_q$ corresponds to the YES case of the $\myk$-partite problem, then instance $\mc I$ admits a schedule of makespan $(\myk+3)Q/3$.
\end{lemma}

\begin{proof}
Suppose that $G_q$ corresponds to a YES instance of the $\myk$-partite problem, and let $V_{i,j}$ for $1\leq i\leq \myk$ and $j <Q$ be the associated partition of the sets $V_i$.
Note that the size of any set $V_{i,j}$ of the partition is at most $(1+Q\varepsilon)n/Q$, since $\sum_{j=0}^{Q-1} |V_{i,j}|=|V_i|=n$ and, by definition, in a YES instance it holds that $|V_{i,j}| \geq \frac{1-\varepsilon}{Q}n$.
We next partition the tasks of \mc I into sets $S_{i,j}$.
For each $z$, $0\leq z < \myk/3$, and $j$, $0 \leq j \leq Q-1$, we define:
\begin{compactitem}
	\item type $A$: $S_{zQ+1,j} = \bigcup_{v\in V_{3z+1,j}} \mc J_{3z+1,v}$, and thus \\$|S_{zQ+1,j}| \leq (Qn-Q)(1+Q\varepsilon)n/Q \leq k (1-1/n)$.
	\item type $B$: $S_{zQ+2,j} = \bigcup_{v\in V_{3z+2,j}} \mc J_{3z+2,v}$, and thus\\ $|S_{zQ+2,j}| \leq Qn^3(1+Q\varepsilon)n/Q \leq (1+Q\varepsilon)n^4 \leq m$.
	\item type $C$: for $1\leq\ell \leq Q-2$, $S_{zQ+2 +\ell,j} = \bigcup_{v\in V_{3z+3,j}}  \{J^\ell_{3z+3,v}\}$, and thus \\$|S_{zQ+2+\ell,j}| = (1+Q\varepsilon)n/Q \leq k/(nQ)$.
\end{compactitem}

Let $\mc T_t$ be the union of all $S_{i,j}$ with $t=i+j$, $1\leq i \leq Q\myk /3$ and $0 \leq j \leq Q-1$.
We create a schedule for instance \mc I as follows: at the time slot $[t-1,t)$, we schedule the tasks of set $\mc T_t$.
A sketch of the beginning of this schedule is given in Table~\ref{tbl:schedule}.
The type and the number of machines (CPUs or GPUs) for executing each set of tasks $S_{i,j}$ is also given in this table.
Note that the tasks of the second triplet $\langle V_4,V_5,V_6 \rangle$ start executing from time slot $[Q,Q+1)$: specifically, $S_{Q+1,0}$ contains tasks in $V_4$.
Moreover, the execution of some tasks of the first triplet $\langle V_1,V_2,V_3 \rangle$ takes place after time $Q+1$: specifically, the last tasks in this triplet belong to the set $S_{Q,Q-1}$ and they are executed in the time slot $[2Q-2,2Q-1)$.
However, there is no a time slot in which 3 triplets are involved.

\begin{table}
\resizebox{\linewidth}{!}{
\begin{tabular}{c||c|c|c|c|c|c|c|c|c}
 & CPU & \multicolumn{8}{c}{GPU} \\
\hline
 & $m$ & $k(1-1/n)$ & $k/(nQ)$ & $k/(nQ)$ & $k/(nQ)$ & $\ldots$ & $k/(nQ)$ & $k/(nQ)$ & $k/(nQ)$ \\
\hline
\hline
$[0,1)$ & & $S_{1,0}$ & & & & & & &\\
$[1,2)$ & $S_{2,0}$ & $S_{1,1}$ & & & & & & \\
$[2,3)$ & $S_{2,1}$ & $S_{1,2}$ & $S_{3,0}$ & & & & & \\
$[3,4)$ & $S_{2,2}$ & $S_{1,3}$ & $S_{4,0}$ & $S_{3,1}$ & & & & \\
$[4,5)$ & $S_{2,3}$ & $S_{1,4}$ & $S_{5,0}$ & $S_{4,1}$ & $S_{3,2}$ & & & & \\
$\ldots$ & $\ldots$ & $\ldots$ & $\ldots$ & $\ldots$ & $\ldots$ & & & & \\
$[Q-1,Q)$ & $S_{2,Q-2}$ & $S_{1,Q-1}$ & $S_{Q,0}$ & $S_{Q-1,1}$ & $S_{Q-2,2}$ & ~$\ldots$~ & $S_{3,Q-3}$ & \\
$[Q,Q+1)$ & $S_{2,Q-1}$ & $S_{Q+1,0}$ & & $S_{Q,1}$ & $S_{Q-1,2}$ & $\ldots$ & $S_{4,Q-3}$ & $S_{3,Q-2}$\\
$[Q+1,Q+2)$ & $S_{Q+2,0}$ & $S_{Q+1,1}$ & & & $S_{Q,2}$ & $\ldots$ & $S_{5,Q-3}$ & $S_{4,Q-2}$ & $S_{3,Q-1}$ \\
$[Q+2,Q+3)$ & $S_{Q+2,1}$ & $S_{Q+1,2}$ & & & & $\ldots$ & $S_{6,Q-3}$ & $S_{5,Q-2}$ & $S_{4,Q-1}$ \\
$[Q+3,Q+4)$ & $S_{Q+2,2}$ & $S_{Q+1,3}$ & $S_{Q+3,0}$ & & & $\ldots$ & $S_{7,Q-3}$ & $S_{6,Q-2}$ & $S_{5,Q-1}$ \\
$[Q+4,Q+5)$ & $S_{Q+2,3}$ & $S_{Q+1,4}$ & $S_{Q+4,0}$ & $S_{Q+3,1}$ & & $\ldots$ & $S_{8,Q-3}$ & $S_{7,Q-2}$ & $S_{6,Q-1}$ \\
$\ldots$ & & \\
\end{tabular}}
\caption{A sketch of the beginning of the schedule for the tasks in \mc I.}
\label{tbl:schedule}
\end{table}

In the last time slot of the created schedule we execute the tasks in $\mc T_t$ with $t=i+j$, $i=Q\myk/3$ and $j=Q-1$.
Hence, the makespan is $Q\myk/3+Q-1 < Q\myk/3+Q$.
It remains to prove the feasibility of the created schedule: the precedence constraints are satisfied and there are enough machines to perform the assigned tasks at each time slot.

Consider first the precedence constraints inside each set $\mc J_{3z+3,v}$, $0\leq z < \myk/3$ and $v \in V_{3z+3}$, that is the arc from the task $J_{3z+3,v}^\ell$ to the task $J_{3z+3,v}^{\ell+1}$, for all $\ell$, $1 \leq \ell \leq Q-3$.
By construction, $J_{3z+3,v}^\ell \in S_{zQ+2+\ell,j}$ and $J_{3z+3,v}^{\ell+1} \in S_{zQ+2+\ell+1,j}$.
Thus, $J_{3z+3,v}^\ell$ is executed in the time slot $zQ+2+\ell+j$, while $J_{3z+3,v}^{\ell+1}$ in the time slot $zQ+2+\ell+1+j>zQ+2+\ell+j$, and hence this kind of precedence constraints are satisfied.

Consider now the precedence constraint from a task $J \in \mc J_{i,v}$ corresponding to $v \in V_{i,j_1} \subset V_i$ to a task $J' \in \mc J_{i+1,w}$ corresponding to $w \in V_{i+1,j_2} \subset V_{i+1}$.
By construction and due to the fact that $G_\myk$ is a YES instance, an arc from $J$ to $J'$ exists only if $j_1\leq j_2$.
Assume that $J$ belongs to the set $S_{i_1,j_1}$, while $J'$ belongs to the set $S_{i_2,j_2}$.
By the definition of the sets $S_{i,j}$, we have that $i_1 < i_2$.
Thus, $i_1+j_1 < i_2+j_2$ which means that $J$ is executed in a time slot before $J'$, and hence this kind of precedence constraints are also satisfied.

It remains to show that each set $\mc T_t$ is composed of at most $m$ CPU tasks and $k$ GPU tasks, so can be computed in a single time slot.
In a given set $\mc T_t$, there can be at most one set of type $A$, one set of type $B$ and $Q-2$ sets of type $C$.
As explained in the definition of the sets $S_{i,j}$, each set of type $B$ is composed of at most $m$ CPU tasks.
Moreover, each set of type $A$ is composed of at most $k(1-1/n)$ GPU tasks, while each of the $Q-2$ sets of type $C$ is composed of at most $k/nQ$ GPU tasks.
In total, there are $k(1-1/n) + (Q-2) k/nQ < k$ GPU tasks, and the lemma follows.
\end{proof}

\begin{lemma}[Soundness]
	\label{lem:sound}
	If $G_q$ corresponds to the NO case of the $\myk$-partite problem,  then
all schedules of instance $\mc I$ have a makespan at least
$f(Q)\myk Q$, where $f$ tends towards 1 when $Q$ grows.
\end{lemma}

\begin{proof}
Suppose that $G_q$ corresponds to a NO instance of the $\myk$-partite
problem, and consider the following partition of the tasks of the
associated instance $\mc I$, for all $0\leq z < \myk/3$:

\begin{itemize}
	\item type $A$: $S_{Qz+1} := \bigcup_{v\in V_{3z+1}} \mc J_{3z+1,v}$, so $|S_{Qz+1}|=Qn^2-Qn = n(n-1)Q$.
	\item type $B$: $S_{Qz+2} := \bigcup_{v\in V_{3z+2}} \mc J_{3z+2,v}$, so $|S_{Qz+1}|=Qn^4$.
	\item type $C$: $S_{Qz+2 +\ell} := \bigcup_{v\in V_{3z+3}}  \{J^\ell_{3z+3,v}\}$, for $1\leq\ell \leq Q-2$, so $|S_{Qz+2+\ell}|=n$.
\end{itemize}

Consider a schedule of $\mc I$ minimizing the makespan, and discard a
fraction $2\delta$ of each set $S_i$ in the partition where $(i\mod Q)\in\{0,1,2,3\}$: the first $\lceil\delta|S_i|\rceil$ tasks to be
executed $S_i^s$ and the last $\lceil\delta|S_i|\rceil$ tasks to be executed $S_i^f$. Let $\mc R$ be the
pseudo-schedule obtained.

Suppose that there exists $i$ such that one task of $S_{i+1}$ is started
before all tasks of $S_i$ are completed, and at least one set among
$S_i$, $S_{i+1}$ is of type $A$ or $B$ ({\it i.e.}, $(i\mod Q)\in\{0,1,2\}$). Then, this means that there is no
edge between $S_i^f$ and $S_{i+1}^s$. Let $i'$ be such that the set
$S_i$ corresponds to vertices of $V_{i'}$, and let $V^f_{i'}$ be the set of vertices $v\in V_{i'}$ that verify $\mc J_{i',v}\cup S_i^f\neq\emptyset$. Define $V^s_{i'+1}$ analogously. 
By the definition of the sets $\mc J$, there is no edge between $V^f_{i'}$ and $V^s_{i'+1}$. As all $\mc J_{i',v}$ have the same size, we have $|V^f_{i'}|\geq |V_{i'}|\cdot|S_i^f|/|S_i|\geq \lfloor\delta n\rfloor$, and, similarly, $|V^s_{i'+1}|\geq\lfloor\delta n\rfloor$. This contradicts the
hypothesis that $G_q$ is a NO instance of the $\myk$-partite problem.
Therefore, in the pseudo-schedule $\mc R$, a task of a set $S_i$ of type $A$ or
$B$ cannot be executed concurrently with a task from the sets $S_{i-1}$ or $S_{i+1}$.

Hence, for any $z$, the set $S_{Qz+1}$ of type $A$  has to be
completed before the start of the set $S_{Qz+2}$ of type $B$, which
itself has to be completed before the start of the set $S_{Qz+3}$ of
type $C$. If $z<\myk/3-1$, the set $S_{Q(z+1)}$ of type $C$ has to be in
turn completed before the start of the set $S_{Q(z+1)+1}$ of type $A$.

Fix $z$ and consider the $Q-2$ sets of type $C$ associated to $V_{3z}$.
Among the $n(Q-2)$ tasks of these sets, at most $2(2\delta n+2) <
n=|V_{3z}|$ have been discarded, so there exists one vertex $v\in
V_{3z}$ for which none of the tasks $J^\ell_{3z,v}$, for $1\leq\ell\leq
Q-2$ has been discarded. Because these tasks form a chain and each task
needs a time 1 to be completed, the pseudo-schedule $R$ needs at least
a time $M_C = Q-2$ to schedule all the sets of type $C$ associated to
$V_{3z}$.

In the pseudo-schedule $\mc R$, if all tasks of a set $S_{Qz+1}$ of type $A$ are executed on
GPUs, this takes a time at least (recall that $n>Q$,
$\varepsilon\leq1/Q^2$ and $\delta\leq1/(2Q)$):
\begin{align*}
M_A &= \frac{|S_{Qz+1}| - |S_{Qz+1}^s|-|S_{Qz+1}^f|}{k}\\
 &\geq \frac{(1-2\delta) n(n-1)Q-2}{(1+Q\varepsilon)n^2+1} \\
&\geq \frac{(1-\frac1Q) (1-\frac 1n)-\frac 2 {Qn^2}}{(1+\frac 1Q)+\frac 1 {n^2}}~ Q \\
&\geq \frac{Q-2}{Q+2}~Q.
\end{align*}

If a task of a set $A$ is executed on a CPU, this takes a time $n>Q\geq M_A$.

A set $S_{Qz+2}$ of type $B$ has to be scheduled on all CPUs and GPUs in time at least:
\begin{align*}
M_B &= \frac{|S_{Qz+2}| - |S_{Qz+2}^s|-|S_{Qz+2}^f|}{m+n\cdot k}\\
 &\geq \frac{(1-2\delta) Qn^4-2}{(1+Q\varepsilon)n^4+1 + n((1+Q\varepsilon)n^2+1)} \\
&\geq \frac{1-\frac 1Q-\frac{2}{Qn^4}}{1+\frac1Q+\frac 1{n^4} + \frac 1n +\frac{1}{nQ} + \frac 1 {n^3} }~ Q \\
&\geq \frac{Q-2}{Q+3}~Q.
\end{align*}

Therefore, the makespan of $\mc R$ is at least:
\begin{align*} 
\frac \myk3\left(M_A + M_B + M_C\right) &\geq 
\frac{\myk Q}{3} ~ \left( \frac{Q-2}{Q+2} + \frac{Q-2}{Q+3} +\frac {Q-2}Q\right). 
\end{align*}

As the expression in parentheses tends towards $3$ when $Q$ grows, and
the makespan of $\mc R$ is not larger than the minimum makespan to
schedule the instance $\mc I$, the lemma holds.
\end{proof}
	
We are now ready to complete the proof.
	
\begin{proof}[Proof (Proof of \Cref{th:LB})]
Let $\alpha>0$ and choose $\myk$ and $Q$ such that $f(Q)\myk Q > (3-\alpha)(\myk+3)Q/3$.
Consider an instance $G_q$ of the corresponding $\myk$-partite problem, with $n>Q$.
Because of \Cref{lem:complete,lem:sound}, if $G_q$ is a YES instance,
then its optimal makespan is at most $(\myk+3)Q/3$, and otherwise, its
makespan is at least $f(Q)\myk Q > (3-\alpha) (\myk+3)Q/3$.

Therefore, an algorithm approximating the scheduling problem within a
factor $3-\alpha$ also solves the $\myk$-partite problem in polynomial
time, which contradicts \Cref{hyp} and  $P\neq NP$.
\end{proof}

We can furthermore adapt this proof to show the following result:

\begin{corollary}
Assuming \Cref{hyp} and $P\neq NP$, the problem \pbname has no $3-\alpha$-approximation, for any $\alpha>0$ and any value of $m/k$.
\end{corollary}

\begin{proof}[Proof (Proof sketch)]
Define CPU tasks as $\ph=1$ and $\pb=\infty$, and GPU tasks as $\ph=\infty$ and $\pb=1$. The value of $k$ is the same as before, but we now consider any value of $m \geq k$, and we define the sets of type $b$ as containing $n_b = \lfloor Qmn/k\rfloor$ tasks instead of $Qn^3$. The completeness lemma is still valid as $(1+Q\varepsilon)n\cdot n_b\leq m$ and the soundness lemma holds as tasks cannot be processed on the other resource type. 
\end{proof}

This result is interesting as the competitive ratio of the algorithms known for \pbname both in the offline ($3+4\sqrt{\frac{1-k/m}{2-k/m}}$) and in the online ($1+2\sqrt{m/k}$~\cite{online2}) setting tend towards 3 when $m/k$ is close to 1, so there is no gap between the conditional lower bound and the upper bound for this case.
Note that this hardness result also holds if an oracle provides the allocation (CPU or GPU for each task), in which case List Scheduling is 3-competitive~\cite[Theorem 7]{online2}. Therefore, the gap between the conditional lower bound and the algorithm \algname is mainly due to the difficulty of the allocation.

%%%%%%%%%%%%%%%%%%%%%%%%%%%%%%%%%%%%%%%%%%%%%%%%%%%%%%%%%%%%%%%%%%%%%%%%%%%%
%%%%%%%%%%%%%%%%%%%%%%%%%%%%%%%%%%%%%%%%%%%%%%%%%%%%%%%%%%%%%%%%%%%%%%%%%%%%
%%%%%%%%%%%%%%%%%%%%%%%%%%%%%%%%%%%%%%%%%%%%%%%%%%%%%%%%%%%%%%%%%%%%%%%%%%%%

\section{Conclusion}
\label{sec:conclusion}

We propose a $(3+2\sqrt{2})$-approximation algorithm \algname for the \pbname problem.
Our algorithm improves the approximation ratio upon the previous 6-approximation algorithm known in the literature, by using a different rounding procedure, which although is not optimal for the allocation phase, leads to a better worst-case ratio for the whole problem.
We also show a conditional lower bound of 3 on the approximation ratio for this problem, assuming a generalized variant of the unique games conjecture, improving over the previous result of 2.
The approximation ratio of \algname actually decreases towards 3 when $m$ and $k$ are close, thus closing the gap with the lower bound for $m=k$.
The natural objective would be to close this gap for all values of $m$ and $k$.

%%%%%%%%%%%%%%%%%%%%%%%%%%%%%%%%%%%%%%%%%%%%%%%%%%%%%%%%%%%%%%%%%%%%%%%%%%%%
%%%%%%%%%%%%%%%%%%%%%%%%%%%%%%%%%%%%%%%%%%%%%%%%%%%%%%%%%%%%%%%%%%%%%%%%%%%%
%%%%%%%%%%%%%%%%%%%%%%%%%%%%%%%%%%%%%%%%%%%%%%%%%%%%%%%%%%%%%%%%%%%%%%%%%%%%

\renewcommand\refname{References} % name for the reference list
{\setstretch{0.8} % linespacing for the references
\addcontentsline{toc}{section}{References} % to change the name of the references in the TOC
\bibliography{biblio.bib} % adds the references to the document
}

% \bibliographystyle{splncs04}
% \bibliography{biblio}

\end{document}